\RequirePackage{fix-cm}
\documentclass[smallextended, natbib]{svjour3}       
\smartqed  
\usepackage{amsmath,amssymb}
\usepackage{graphicx}
\usepackage{enumerate}
\usepackage{subfig}
\usepackage{url}

\begin{document}

\title{Reciprocating Preferences Stablize Matching: \\ College Admissions Revisited
}


\author{Jerry Jian Liu         \and
        Dah Ming Chiu 
}


\institute{J. Liu \at
              Department of Information Engineering, The Chinese University of
Hong Kong, Shatin, NT, Hong Kong. \\
              Tel.: +852-3163-4296\\
              Fax: +852-2603-5032\\
              \email{liujtm@gmail.com}           
           \and
           D. Chiu \at
              \email{dmchiu@ie.cuhk.edu.hk}
}


\maketitle

\begin{abstract}
In considering the college admissions problem, almost fifty years ago, Gale and Shapley came up with
a simple abstraction based on preferences of students and colleges. They introduced the concept of
stability and optimality; and proposed the \emph{deferred acceptance} (DA) algorithm that is proven
to lead to a stable and optimal solution. This algorithm is simple and computationally efficient.
Furthermore, in subsequent studies it is shown that the DA algorithm is also strategy-proof, which
means, when the algorithm is played out as a mechanism for matching two sides (e.g. colleges and
students), the parties (colleges or students) have no incentives to act other than according to
their true preferences. Yet, in practical  college admission systems, the DA algorithm is often not
adopted. Instead, an algorithm known as the \emph{Boston Mechanism} (BM) or its variants are widely
adopted. In BM, colleges accept students without deferral (considering other colleges' decisions),
which is exactly the opposite of Gale-Shapley's DA algorithm. To explain and rationalize this
reality, we introduce the notion of \emph{reciprocating preference} to capture the influence of a
student's interest on a college's decision. This model is inspired by the actual mechanism used to
match students to universities in Hong Kong. The notion of reciprocating preference defines a class
of matching algorithms, allowing different degrees of reciprocating preferences by the students and
colleges. DA and BM are but two extreme cases (with zero and a hundred percent reciprocation) of this
set. This model extends the notion of stability and optimality as well. As in Gale-Shapley's original
paper, we discuss how the analogy can be carried over to the stable marriage problem, thus
demonstrating the model's general applicability.
\end{abstract}

\section{Introduction}
One of the main features of many market and social processes is their bilateral structure and the
need to match agents from one side of the market with the other side, e.g. students and schools in
college admissions, employees and companies in the job market \citep[see][]{JobMarket}, men and women in
online dating sites, advertisers and advertising slots in sponsored search \citep[see][]{SponSearch}. A class
of ``two-sided matching model'' for studying such problems was first introduced by Gale and Shapley
in 1962 in their seminal paper \citep{GS}, in the context of college admissions and the marriage
problem. Yet, the original model is quite general and can be easily adapted for other two-sided
markets, such as the well-known National Resident Matching Program (NRMP)\footnote{Readers can refer to the official website \url{http://www.nrmp.org/} for details.} for assigning
medical students to residency positions in US.

In Gale and Shapley's seminal paper, the college admissions problem is formulated as follows. Let
there be a set of colleges $\mathcal{C}=\{c_1,\ldots,c_m\}$ and a set of students
$\mathcal{S}=\{s_1,\ldots,s_n\}$. Each college $c_i$ ($i\in\{1,\ldots,m\}$) has a quota $q_i$, that
denotes the maximum number of students it can admit. Each student $s_j$ ($j\in\{1,\ldots,n\}$) can
apply to any number of colleges, represented by a \emph{strict} preference (ordering) over the $m$
colleges.\footnote{ By {\em strict} preference we mean that a student is NOT indifferent between any
two colleges. This simplifies our discussion. In practice, ties can be broken by random lotteries. }
Each college also has a strict ranking for all students (e.g. according to their test scores,
interview performances or other criteria). A \emph{matching} is simply an assignment of students to
colleges such that each college accepts no more students than its quota, and each student is admitted
to at most one college. For convenience, we use $c_0$ to denote a dummy college that takes all
unmatched students. The college admissions problem can then be stated as the problem of designing an
algorithm to arrive at a matching satisfying certain properties. The \emph{mechanism} that implements
the solution can be considered as a blackbox that takes the students' and colleges' preferences
(denoted by $p$) as input, and outputs the matching, as depicted in
Fig. \ref{mechanism_prev}.\footnote{There are more complicated mechanisms that allow students and
colleges to interact with each other multiple times before determining the matching, e.g. in
\cite{School_choice}. They are not considered in this paper.}

In \cite{GS}, two desirable properties are defined: \emph{stability} and \emph{optimality}. Stability
of a matching means that it is not possible to find a student and a college that are not matched to
each other, but both prefer each other more than their current match. In college admissions,
optimality means student-optimality. A stable matching is \emph{student-optimal} if every student is
at least as well-off as he/she is in any other stable matching. College-optimality can be defined in
a similar way, but it is never used. Gale and Shapley proved that there always exists stable
matchings. Their proof is by construction (using the deferred acceptance (DA) algorithm in
Fig. \ref{mechanism_prev}). In the special case when the colleges' preferences are all based on
student examination results, the DA algorithm is equivalent to a procedure that sorts all students
according to merit and letting students choose the most preferred college (if quota still allows) in
order, starting from the best student. It is also shown in \cite{GS} that the optimal stable matching
is unique.

\begin{figure}[hbt]
  \centering
  \includegraphics[scale=0.7]{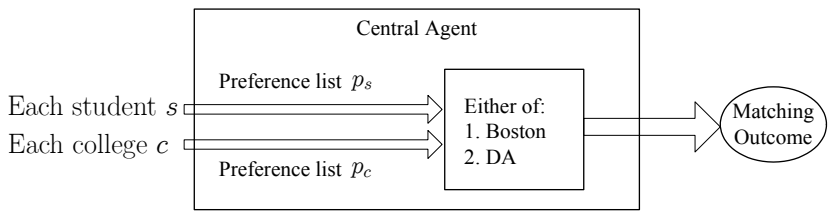}\\
  \caption{Classic Matching Mechanism}\label{mechanism_prev}
\end{figure}

In practice, it is reportedly quite common that the Boston Mechanism (BM) \citep[see][]{Boston} is used
instead. The basic procedures of BM are as follows: in the first round each college $c$ only
considers students who listed it as their first choice and accepts these students one by one
according to $c$'s own preference until the quota of the college is filled up.\footnote{Under BM, all
the offers made so far are \emph{committed} and cannot be changed in subsequent rounds. In another
word, colleges accept students \emph{without deferral}. In contrast, under DA, colleges accept
students \emph{tentatively} in each round and would reconsider all applicants in later rounds.
Matching results under DA cannot be determined until all the iterations are over.} In the second
round, only the remaining student without any offers and colleges still having unfilled quota are
considered. Each remaining college then considers students who listed it as the second choice and
assigns offers to them one by one until the quota is exhausted. The process goes on round by round
until all quota are filled up or all students are already accommodated. One major criticism
concerning BM is its lack of stability. Judged by classic matching theory, matching results under BM
are far from stable and thus many students are ``incentivized'' to make complicated strategies when
submitting their preference lists. In contrast, results under DA are both stable and optimal for
students. Students can feel free to reveal their true preference lists. The transition from BM to DA
is therefore suggested in literature such as \citep{School_choice,Chen,Games}, which arguably would
lead to significant efficiency gains for the whole community. However, our investigation of some practical college admissions systems shows that \emph{certain hybrid mechanism may be more acceptable in society from the perspective of colleges' enrollment concerns as well as students' personal interest}.

In this paper, we propose a generalized model for college admissions, which considers the tradeoff
between students' eligibility and interest by adjusting an additional parameter called
\emph{reciprocating factor}. The larger the reciprocating factor is, the more would the interest
factor counts when inspecting the applicants. BM and DA are merely two special cases of the
generalized model when setting different reciprocating factors. Our model also extends the classic
notion of stability and optimality by re-examining the formation of agents' preferences.

The remainder of the paper is structured as follows. In Section \ref{sec:CA_HK} we elaborate the
actual college admissions system used in Hong Kong which inspires our formal model. In Section
\ref{sec:gene_model} we proposes the generalized model for stable matching. Some of the important
properties of generalized model are presented and discussed extensively in Section
\ref{sec:property}. Section \ref{sec:simu} compares the efficiency of generalized mechanism with the
classic Gale-Shapley mechanism through simulation. We give some further discussion and related work
in Section \ref{sec:discussion} and conclude in Section \ref{sec:conclusion}.

\section{College Admissions in Hong Kong}\label{sec:CA_HK}
The Joint University Programmes Admissions System (JUPAS) is the central system for students to apply
to the nine participating tertiary institutions in Hong Kong \citep{JUPAS}. In JUPAS each student can
apply to at most 25 programmes in order of preference. These preferences are sub-divided and made
known to the institutions in the form of five bands as shown in Table 1.
\begin{table}[htbp]
  \centering
 \caption{Correspondence between Applicants' Band Order and Actual Choice Order}
  \begin{tabular}{c c}
      BAND & \; PROGRAMME CHOICE NO.\\ \hline
      A & 1 - 3\\
      B & 4 - 6\\
      C & 7 - 10\\
      D & 11 - 14\\
      E & 15 - 25\\
  \end{tabular}
\end{table}

The band order is made known to the programmes, however, the inner-band preferences are unrevealed.
For example, in band A, the programme has no idea whether a student lists it as his/her first, second
or third choice.

After aggregating the preference lists from all the applicants, each programme will make a
``\emph{merit order list}''  for its applicants in accordance with its criterion for selection. The
rating criterion is determined independently by each programme: although many programmes would adopt
a Boston-like criterion which assigns band A students with highest priority, some programmes may also
rate students only by their eligibility.\footnote{Student eligibility is judged based on their
academic performances, interview performances and extracurricular activities jointly. Examination
score, which reflects their academic performances, is a dominant factor in determining the
eligibility. \citep{JUPAS}} Some unpopular programmes tend to employ the latter strategy if they find
most excellent students have listed it as band B or band C choices rather than band A.

Finally, after all the \emph{merit order list}s and the applicants preference lists are sent to the
JUPAS office, a central computer will automatically match the applicants with appropriate programmes.
The matching process applies the classic \emph{student-proposing DA algorithm} to give the students
the best offer he/she can possibly obtain.

Although the DA algorithm is used in the last phase, the JUPAS mechanism as a whole is not equivalent
to the classic \emph{Gale-Shapley student optimal mechanism} in \cite{GS}, where students would
truthfully reveal their preference lists.\footnote{Henceforth we still use the term ``DA'' to denote the pure DA mechanism, i.e., the classic Gale-Shapley student optimal mechanism. Otherwise, we will use the term ``DA algorithm'' explicitly when illustrating the JUPAS-like hybrid mechanism.} Applicants in JUPAS face a similar problem like students in BM: their band A choices would receive higher priority than choices in other bands, although there
are no discrimination over the multiple inner-band choices. Actually in JUPAS the applicants are
always advised to choose appropriate programmes according to their interests as well as their
qualifications \citep{JUPAS}. It is never a dominant strategy for students to always reveal their true
preferences. Therefore JUPAS may be interpreted as a hybrid of BM regarding inter-band discrimination
and DA regarding the algorithm applied in the final phase. The statement of ``Gale-Shapley student
optimal stable mechanism is used in Hong Kong'' by \cite{School_choice} is rather mis-informed and misleading: readers may falsely assume that
applicants in Hong Kong could feel free to write down their true preferences, whereas in fact there
is still room for students to manipulate their preference lists.

A natural question is why not replace JUPAS with DA as suggested in most literature. To answer it, we formulate a model in the next section in order to justify such selection of the policy makers for sticking with the hybrid mechanism.

\section{Reciprocating Preference: A Generalized Model of Stable Matching}\label{sec:gene_model}
We have mentioned in the previous section that programmes under JUPAS in Hong Kong have full right to
determine how to rate students. Two factors are the most important: student eligibility and band
order. In practice, most programmes put heavy weight on academic performance in determining students'
eligibility which makes the examination scores a very decisive factor in admission. To simplify the
analysis, we assume each applicant would attend a standard examination and gain a total score which
ranges from zero to the maximum mark. For band order, higher preference (like Band \emph{A})
indicates that the applicant is more interested in the programme while ranking certain programme as
Band \emph{D} or \emph{E} infers the lack of interest.

Let $\mathcal S$ denote the set of applicants, $\mathcal C$ denote the set of programmes.\footnote{In
our model we take ``programmes'' and ``colleges'' equivalently, both as the counterpart of
students/applicants.} Each student $s\in \mathcal S$ achieves a total score $f_s\in [0, f_{max}]$ in
the standard examination where $f_{max}>0$ is the full mark of the examination. Each programme $c\in
\mathcal C$ has a quota $q_c$. When student $s$ applies to programme $c$ as his/her $r$-th choice,
$s$ will obtain a bonus score $h_c(r)$ which would promote his/her position in programme $c$'s {\em
merit order list}. Generally, the bonus score should be a strictly decreasing function over
\emph{preference order/ranking} $r$. That is to say, the smaller $r$ is, the more bonus score it would bring.
In practice, the programme director could make a corresponding table mapping each \emph{preference
order} to a certain bonus score for ease of reference. All applicants are then sorted by their
\emph{merit scores} in each programme, where student $s$'s \emph{merit score} in programme $c$ is
computed according to the following equation:
\begin{equation*}
    {mrt}_c(s)=(1-\alpha_c)\cdot f_s+\alpha_c \cdot h_c(r),\quad \alpha_c\in [0,1]
\end{equation*}
The first term denotes the original score achieved by $s$ and the second term is the bonus score for
students' interest in $c$. In case of tie when students share the same merit score, $f_s$ serves as
the tie-breaker and student with higher $f_s$ obtains higher priority. Finally, if all terms are
equal, we break the tie by a random lottery.

We refer to $\alpha_c$ as the {\em reciprocating factor} (RF), a constant determined independently by
each programme $c$, reflecting its sensitivity towards applicant's {\em preference order}. Programmes
with larger RF place more weight on applicants' personal interests: other things being equal,
students whose interests match with the programme are more favored. In the extreme case when $\alpha$
of different programmes all equal to zero, the matching reduces to exactly the Gale-Shapley student
optimal mechanism: preference order would not affect students' positions in programmes; in the
contrary, when all $\alpha$ are set to one,\footnote{To be more accurate, besides the situation when all $\alpha_c=1$, there are actually an infinitely large amount of pairs $\alpha_c$ and $h_c$ which can implement the Boston mechanism, as long as it holds that for any $c$ and $r$, $h_c(r-1)-h_c(r)>\frac{1-\alpha_c}{\alpha_c}f_{max}$ and $\alpha_c\in(0,1)$. We thank an anonymous referee for pointing this out.} the matching works in the same way as the Boston
mechanism: the first choice gets the highest priority. For a general $\alpha_c\in (0,1)$, say
$\alpha_c$ equals to $0.2$, it means that programme $c$ would count $80\%$ of original score and
$20\%$ of interest factor when evaluating the applicants.

After calculating the merit score for each applicant, each programme could generate a
\emph{reciprocating preference} list by comparing merits scores of the applicants. These
reciprocating preferences are then sent to the central college admissions system for further
processing. The complete procedures can be illustrated in Fig. \ref{mechanism_ext} where students'
\emph{reciprocating preferences} remain the same as their original preferences $p_s$.

\begin{figure}[hbt]
  \centering
  \includegraphics[scale=0.65]{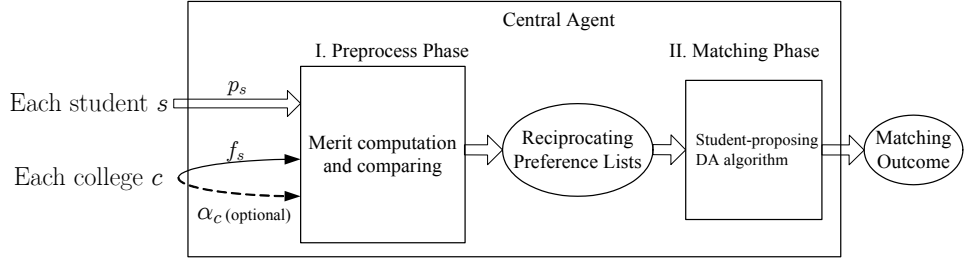}\\
  \caption{Generalized Matching Model}\label{mechanism_ext}
\end{figure}

The reciprocating factor we propose in this paper gives programmes more flexibility in choosing a
``reasonable'' enrollment mechanism:
\begin{itemize}
  \item For programmes which hope to stick to the traditional Boston-like scheme, there will be no need
for any change since by default $\alpha$ is set to one;
  \item For programmes whose sole objective is to raise the average score of newly admitted
  students, setting $\alpha$ to zero would be their favorite strategy;
  \item For other elastic programmes concerning the students' interest as well, a suitable
  $\alpha$ between zero and one needs to be determined according to each programme's own admission policy in each admission year.
\end{itemize}

\section{Properties of Generalized Stable Matching}\label{sec:property}
In last section we proposed the notion of \emph{reciprocating preference} which can better reflect
the selection criteria of individual programmes. It it the student with higher merit score, rather
than higher original exam score, that is more favored by each programme. With the change of the
interpretation of agents' preferences, important concepts such as stability and optimality
should also need to be re-defined \emph{from the perspective of reciprocating preferences}. We now give the formal definition as follows.

\begin{definition}
A matching is \emph{R-stable}\footnote{Similarly, we can use the term ``R-stability'' or ``reciprocating-stability'' in the noun form.} if it is not possible to find a student and a college that are not matched to each other, but both \emph{prefer} each other more than their current match when judged by their \emph{reciprocating preferences}.
\end{definition}

A fundamental property of the generalized mechanism can be then presented as follows.

\begin{property}\label{theo:R_stable}
Matching outcome under the generalized mechanism is \emph{R-stable}. Moreover, it is \emph{optimal} among all possible R-stable outcomes.
\end{property}

The above theorem is easy to see by ignoring the preprocess phase in Fig. \ref{mechanism_ext} and
focusing on the matching phase where DA always generates stable and optimal matching given any
preference lists. Therefore from the perspective of reciprocating preferences, the matching outcome
under JUPAS-like hybrid mechanism is still stable and optimal among all stable outcomes.

Since BM is merely an extreme case of the generalized mechanism, we may easily get the following
corollary through Property \ref{theo:R_stable}.
\begin{corollary}
The Boston mechanism, known to be unstable by classic matching theory
\citep[see][]{School_choice,Chen,Games}, is actually R-stable with regard to colleges' reciprocating preferences.
\end{corollary}

In previous analysis, we assume all students would submit their true preferences. However, this
assumption is unrealistic, especially in practical college admissions system like JUPAS where
strategic behaviors is actually quite common. We now show some key results for strategy analysis as follows.

\begin{property}\label{theo:no_truthful_mech}
College admissions mechanism like JUPAS is NOT strategy-proof in general. The exceptions are when
all the programmes' reciprocating factors are zero, it is the dominant strategy for students
to reveal their true preferences.
\end{property}
\begin{proof}
We defer the detailed proof to Appendix \ref{proof:no_truthful_mech}. \qed
\end{proof}

When all programmes' reciprocating factors are zero, the generalized model reduces to the classic
Gale-Shapley student-optimal mechanism which is strategy-proof for students \citep[see][]{Dubins,Roth82}.
Generally, when most programmes have positive reciprocating factors, students would act
strategically: students with relatively lower examination scores may try to avoid some popular
programmes where there would be lots of competitors with higher scores, and list those unpopular ones
as their top choices instead in order to obtain more bonus scores there and increase the chances for
admission.

It is worth pointing out that even though the generalized model is not strategy-proof, it does not
imply that students would always have strong incentive to strategize. Whether applicants' strategies
can work or not depends on how much information they know about other students' behaviors. For
example, suppose there are five students $s_1,s_2,\ldots,s_5$ and $s_5$ gets the lowest exam score;
there are two programmes $c_1$ and $c_2$, both with reciprocating factors $\alpha=1$ and quota $q=1$.
If the preferences of students are strongly \emph{correlated}, say all students consent that $c_1$ is
better than $c_2$, $s_5$ may try to avoid the popular programme $c_1$ and list $c_2$ as the first
choice to get higher priority. Otherwise, $s_5$ may end up with no offers at all. However, if
students' preferences are totally \emph{uncorrelated}, which means different students hold
independent views to the programmes, $s_5$ may know little about others' preferences.\footnote{In
typical college admissions system like JUPAS, students' submitted preference lists are private
information. Applicant cannot access to other students' preference lists.} In this scenario telling
truth might be the best choice of students. The following property state this observation formally.

\begin{property}\label{theo:Bayes_NE}
Suppose all programmes have equal quota and their reciprocating factors are independently drawn from a uniform distribution over any range $[a,b]$ ($0\leq a<b\leq 1$). Besides, the preferences of
all students are independently drawn from a uniform distribution over the set of all possible rank
orderings (i.e., the uncorrelated environment). Then, truth-telling is a Bayesian Nash equilibrium
for each student under the generalized matching mechanism.
\end{property}

\begin{proof}
For readability, we defer the detailed proof to Appendix \ref{proof:Bayes_NE}.  \qed
\end{proof}

Property \ref{theo:Bayes_NE} shows that students would be forced to reveal their true preferences if they have no exact knowledge of other students' realized preferences except the distribution. This is one extreme case with totally uncorrelated preferences. At the other extreme is the totally correlated case where students have complete information of other students' true preferences. Property \ref{theo:dominated} characterizes the possibilities of students' strategies in the latter scenario.

\begin{property}\label{theo:dominated}
Suppose all students have the same preference over colleges. Besides, this is a common knowledge among all students. Let the quota of their favorite college be $q$. Then truth-telling is a \emph{dominated strategy} for all students except the top $q$ students ranked according to their scores.
\end{property}
\begin{proof}
We present the proof in Appendix \ref{proof:dominated}.  \qed
\end{proof}

In practice, students can only have partial knowledge to other students' submitted preferences, probably through the history of admission data in each individual programme. Besides, students need to estimate the value of $\alpha_c$ in each programme and carefully choose programmes they make strategy on since promoting preference in purpose for programmes which run DA-like mechanism (i.e., $\alpha$ is slightly above zero) may have little effect in raising their merit scores in those programmes. On the other hand, in each admission year programmes would take lessons from admission outcomes of the previous year and adjust their policy on how to determine parameters like $\alpha$ and $h$. These interactions form a complicated extensive-form game which may be repeated infinitely\footnote{Readers should be alert that this may not be construed as the ``infinitely repeated game'' in the standard term of game theory where the set of participants in the game remains unchanged typically. However, in the college admission settings, both the policy makers of the programmes and applicants may be different for each admission year.} at the time granularity of each admission year. This leaves a bunch of open questions to be settled such as: whether there are any equilibrium strategies and if they do exist, would the matching outcomes converge to any of these equilibrium states after finite rounds of games year by year.

Until now we have focused on students' strategies in submitting preferences and assumed that college programmes would reveal their reciprocating factors nonstrategically. While it seems intractable to characterize the dynamic game in the long run, we do obtain some analytical result for the single shot game in one particular admission year. The following property helps relieve our concerns of programmes' strategies so that we may refocus on the manipulations from applicants' side.

\begin{property}\label{theo:college_truth}
Suppose after students submit their preferences, each college $c$ determines its true values of $\alpha_c$ and $h_c(r)$ independently. Then revealing $\alpha_c$ and $h_c(r)$ truthfully would be the dominant strategy for any programme $c\in\mathcal{C}$ in the \emph{ex post} perspective.
\end{property}
\begin{proof}
We present the proof in Appendix \ref{proof:college_truth}.  \qed
\end{proof}

Here we take the ex post perspective\footnote{For the \emph{ex ante} case, say colleges need to submit their parameters before students' submission, Property \ref{theo:college_truth} may not hold since colleges would take students' reaction into consideration before making any decisions. It would be an interesting future work for analyzing the complicated interactions from the ex ante perspective.} since when colleges make their strategies, students' preferences are already submitted and fixed. In practice, when the admission authority evaluate the matching outcome each year, they can only conduct it based on the submitted preference lists of students. It would be expensive (or even impossible) to obtain the true preferences of all applicants via survey\footnote{Still, students have no incentive to reveal their true preference to authority after the matching is over.} or other methods. Thus for the single shot game, colleges would passively reveal their true parameters once they receive applicants' preference lists.

In this section we have shown some fundamental properties of the generalized mechanism. To further investigate the degree of satisfaction participants perceive under different mechanisms, we implement simulations and present the social welfare results different mechanisms induce in the next section.

\section{Performance Evaluation for Generalized Matching Mechanism}\label{sec:simu}
To define and compare social welfare under different matching outcomes, we need to quantify the
utility of each participant/agent in the mechanism first.

Denoted by $\mathcal{S}$ the set of students and $\mathcal C$ the set of colleges.
$\mathcal{I}=\mathcal{S}\cup \mathcal{C}$ is the set of all participants and $\mathcal {O}$ is the
set of all possible matching outcomes. For any agent $i\in \mathcal I$, let $p_i$ be the
reciprocating preference list\footnote{Here when considering students, $p_i$ is students' original preference.} of $i$ and ${o}(i)$ be the set of participants matched to $i$ under
certain outcome ${o}\in \mathcal{O}$.

For each agent $j\in o(i)$, denote integer $r(j,p_i)$ as the ranking agent $j$ appears in $i$'s
preference list $p_i$. For example, $r(j,p_i)=1$ means that $j$ is the first choice in $i$'s
preference list.\footnote{For simplicity, we assume there are no ties in the preference list.
Otherwise, we can break the tie by a random lottery first.} We use $r=0$ to denote the unmatched case.

We assume that agent $i$'s utility is additive and only determined by the orders of the matching set in the
preference list, which can be written as,
\begin{equation}\label{eq:utility}
    u_i(o)=\sum_{j\in o(i), j\in \mathcal{I}} U_i(r(j,P_i)) \qquad \forall i\in \mathcal{I}, o\in\mathcal{O}
\end{equation}
where $U_i(r)$ is non-increasing as integer $r (r>0)$ increases. Intuitively, it means that higher order
(smaller $r$) would generate higher degree of satisfaction for agent $i$.

The aggregate utility of students (or colleges) under outcome $o$ is then:
\begin{eqnarray*}
  \pi_S(o) &=& \sum_{s\in \mathcal{S}} u_s(o), \qquad \forall o\in\mathcal{O}; \\
  \pi_C(o) &=& \sum_{c\in \mathcal{C}} u_c(o), \qquad \forall o\in\mathcal{O}.
\end{eqnarray*}

The \emph{social welfare} is defined as the aggregate utility of all participants in the mechanism,
which can be written as follows:
\begin{equation*}
    \Pi(o)=\sum_{i\in I} u_i(o)=\pi_S(o)+\pi_C(o), \qquad \forall o\in\mathcal{O}.
\end{equation*}

We say matching outcome $o_1$ is more \emph{efficient} than $o_2$ if:
$$\Pi(o_1)>\Pi(o_2) \qquad o_1,o_2\in\mathcal O.$$

We further say mechanism $\mathcal {M}_1$ is more \emph{efficient} than $\mathcal {M}_2$ if $\mathcal
{M}_1$ can always induce a more efficient matching outcome than $\mathcal {M}_2$ under any possible
preference lists of agents. Generally speaking, an outcome would be more efficient if it induces
higher ranked matching. In the context of college admissions, a mechanism which generates more
first-choice matching for students is likely to be more efficient.

\subsection{Simulation Settings}
Suppose there are $10$ students and $5$ colleges with just one quota in each college. We assume the
utility of each agent is as follows,
\begin{eqnarray*}
  U_s(r) &=& 11-r,\qquad r\in\{1,2,\ldots,5\} \\
  U_c(r) &=& 11-r,\qquad r\in\{1,2,\ldots,10\} \\
  U_s(0) &=& U_c(0)=0.
\end{eqnarray*}
Thus a first-ranked matching would bring in utility of $10$ for either students or colleges. Notice
that the efficiency upper bound is $100$ since there are at most five pairs of students and colleges
matched with each other.

The preference lists of students are generated as follows:

Student $s$ evaluates each college $c$ by this formulae,
$$g_s^c=\beta g^c+(1-\beta) g_s(c),\qquad \beta\in [0,1]$$ where
$(g^{c_1},g^{c_2},g^{c_3},g^{c_4},g^{c_5})=(100,90,80,70,60)$ denotes the social reputation of each
college and $g_s(c)$ denotes the individual preference of student $s$, which is independently drawn
from uniform distribution over $[0,100]$. The factor $\beta$ denotes the \emph{degree of correlation}
for students' preferences.\footnote{$\beta$ can be also regarded as a signal implying how much information students may know about other students' preferences. If $\beta=1$, each student has complete information to others' preferences since they share exactly the same opinion over colleges. If $\beta=0$, each student has independent opinion over colleges and thus has no posterior knowledge on others' preferences. In general, students' individual opinion would be more or less be affected by the common social opinion and each of them would have partial knowledge to other students' preferences.} The preference list $p_s$ can therefore be deduced by comparing the value
of $g_s^c$ for different $c$, i.e., $s$ prefers $c_1$ to $c_2$ if $g_s^{c_1}>g_s^{c_2}$.

The preference lists of colleges are generated as follows:

College $c$ evaluates each student $s$ by this formulae,
$$mrt_c(s)=(1-\alpha_c)\cdot f_s+\alpha_c \cdot h_c(r)$$
where $h_c(r)=110-10r$ if $c$ is the $r$-th choice in $p_s$. $f_s$ is the exam score of student $s$,
which is independently drawn from uniform distribution over $[0,100]$. We generate the reciprocating
factor $\alpha_c$ for each college $c$ by the following distribution:
$$\alpha_c=\begin{cases}
0 \qquad prob=1/2 \\
1 \qquad prob=1/2
\end{cases}$$

The preference list $p_c$ can then be inferred by comparing the value of $mrt_c(s)$ for different
$s$. In case of ties, namely, $mrt_c(s_1)=mrt_c(s_2)$, $s_1$ is favored over $s_2$ if
$f_{s_1}>f_{s_2}$.

\subsection{Efficiency Results}
After generating the reciprocating preference lists for both sides, the matching outcome can be
obtained by applying the student-proposing DA algorithm. We then calculate the social
welfare under the matching outcome.

In the simulation, we use $\beta\in[0,1]$, with a step size of $0.01$. For each particular $\beta$,
we repeat the process of preference generation for $1000$ times and compute the average values of
aggregate utility and social welfare. For comparison, we also calculate the average social welfare
under pure GS mechanism, which can be easily implemented by just setting $\alpha_c\equiv 0$ for each
college $c$ in the distribution of reciprocating factors.\footnote{Here in GS mechanism we let the true preferences of colleges be endogenized via students' exam scores only, since this is how the admission authority evaluates the social welfare of DA in practice. (One such example is that schools' welfare is measured by the average score of their admitted students in equation (1) of \cite{ChiuWeng}.) Thus our simulation result can serve as a predictor for comparing the official evaluation results, if there are any, of DA and hybrid mechanism released by corresponding admission authorities. Nevertheless, if we accept that reciprocating preferences are the ``true'' preferences of colleges and regard submitting $\alpha_c=0$ as certain manipulation to the true preferences, then our comparison of social welfare of DA and hybrid mechanism should have been based on the same (true) preferences of colleges. Since the latter kind of comparison is already fully covered  by Property \ref{theo:college_truth} theoretically, we feel that there would be no need to conduct further simulations on it (which would not produce any ``surprising'' results beyond theory).}

\begin{figure}[htb]
  \centering
  \includegraphics[width=0.6\textwidth]{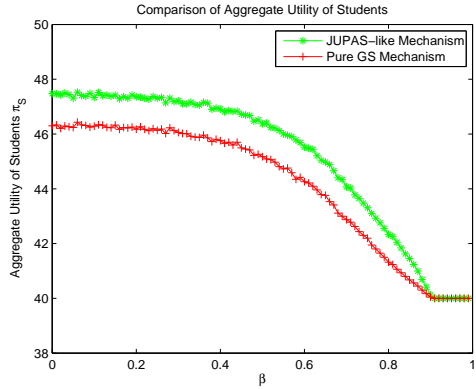}\\
  \caption{Expected Aggregate Utility of Students over Different Degrees of Preference Correlation}\label{Uti_Student}
\end{figure}

Fig. \ref{Uti_Student} presents the simulation results for aggregate utility of students over
different degrees of preference correlation. As we can see, the expected aggregate utilities under
both mechanisms decrease as $\beta$ increases from zero to one. The upper bound of $\pi_S$ is 50
since there are at most five students who can receive their first-choice offers from colleges. When
$\beta$ is small (less than around $0.4$), we can achieve about $94\%$ and $92\%$ of the upper bound
under the JUPAS-like hybrid mechanism and the pure Gale-Shapley (GS) student optimal mechanism
respectively. As $\beta$ rises, the preference list of each student becomes more and more similar and
there are more collision between students' interest in colleges. When $\beta$ is large enough
(greater than $0.91$), the pre-determined social reputation of each college becomes the dominant
factor in forming the preference lists of students. That is to say, $P_s$ would be
$c_1>c_2>c_3>c_4>c_5$ for all students. Thus college $c_1$ would always bring utility of $10$ to the
student community, $c_2$ brings $9$ and so on, which forms this lower bound of
$\pi_S^{low}=10+9+8+7+6=40$.

We also notice in Fig. \ref{Uti_Student} that in general $\pi_S$ is slightly larger under the hybrid
mechanism than under the GS mechanism. This result helps ease the concern that the JUPAS-like
mechanism would hurt the interest of student community as a whole. The intuition is that while some
students with higher exam scores may get worse in the hybrid mechanism, other students with slightly
lower scores would have more chances to enter the programmes/colleges in which they are really
interested.

\begin{figure}[htb]
  \centering
  \includegraphics[width=0.6\textwidth]{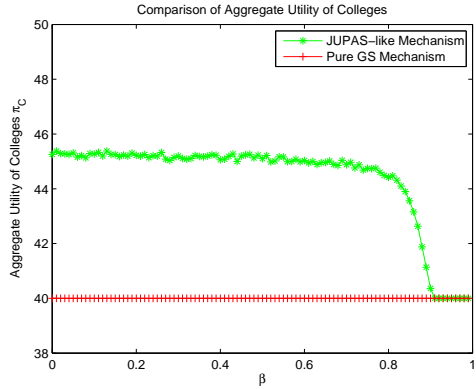}\\
  \caption{Expected Aggregate Utility of Colleges over Different Degrees of Preference Correlation}\label{Uti_College}
\end{figure}

Fig. \ref{Uti_College} shows the result for aggregate utility of colleges. In the pure GS mechanism,
since preferences of colleges are only determined by the exam scores of students, all colleges would
share exactly the same preference list over students. Therefore the student with the highest score
would always bring utility of $10$ to the college side, the student with the second highest score
brings $9$ and so on. That is why $\pi_C$ would be always equal to $10+9+8+7+6=40$ under the GS
mechanism. The upper bound of $\pi_C$ is also 50, which occurs only if all five colleges realize
their first choices. As shown in the figure, we can achieve about $93\%$ of the upper bound under the
JUPAS-like hybrid mechanism when $\beta\in[0,0.8]$. As students' preferences become more similar,
colleges tend to have similar reciprocating preference, which means more conflict would occur among
different colleges. Thus as $\beta$ continues increasing from about $0.8$, the aggregate utility of
colleges would decrease rapidly. When $\beta$ is large enough (greater than $0.91$), all colleges
would share the same preference over students. Thus student with the highest score would always bring
utility of $10$ to colleges, student with the second highest score brings $9$ and so on, which forms
the lower bound of $\pi_C^{low}=10+9+8+7+6=40$.

\begin{figure}[htb]
  \centering
  \includegraphics[width=0.6\textwidth]{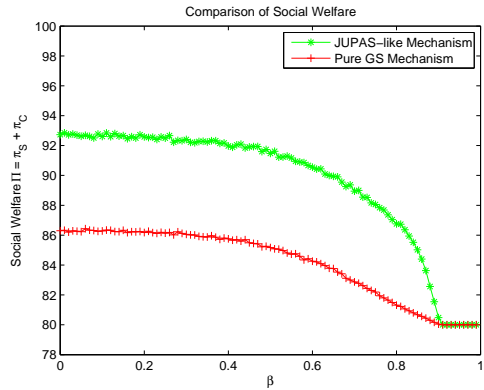}\\
  \caption{The Expected Social Welfare over Different Degrees of Preference Correlation}\label{social_welfare}
\end{figure}

The expected social welfare under different values of $\beta$ is shown in Fig. \ref{social_welfare}.
In the same way, we obtain the upper bound of social welfare as $\Pi^{up}=100$. When $\beta$ is small
(less than $0.5$) and students have various preferences over colleges, we achieve about $93\%$ and
$86\%$ of the upper bound under the hybrid mechanism and the GS mechanism respectively. When $\beta$
approaches to one and students share common opinion on colleges, the ratio would both decrease to
$80\%$. This comparative result of social welfare helps justify the implementation of JUPAS-like
mechanism in college admissions. The transfer from the hybrid mechanism to the GS mechanism can only
achieve the well-known incentive compatible property at the cost of potentially significant loss of
efficiency, especially when students have independent opinions on different colleges.

\subsection{Strategy Issues}
As we have mentioned in the last section, when students' preferences become more correlated, it would
be easier for strategic students to manipulate the matching results successfully. Thus in the
following simulation we use $\beta=1$ as an example to illustrate the effect of possible strategies
of students.

When considering $\beta=1$, the true preferences of students are all the same as $p_s:
c_1>c_2>c_3>c_4>c_5$. We assume there are two types of students: the truthful ones and the strategic
ones. The truthful students would always submit their true preferences while the strategic students
are likely to manipulate their submitted preference lists based on the information they already
collect. In our simulation setting, the key information strategic students may infer is the
reciprocating factor of each college. Notice that in practical college admissions mechanism like
JUPAS, students would know their own examination scores and overall performances of other students
before submitting their applications. Here in the simulation we assume the strategic students would
know exactly their score ranking out of all students. Then one possible strategy would be as follows:

\begin{equation*}\mbox{Strategy $\mathbb{S}$: submit}
\begin{cases}p_s^{(0)}: c_1>c_2>c_3>c_4>c_5 & \mbox{$f_s$ is the highest score;}
\\
p_s^{(1)}: c_2>c_3>c_4>c_5>c_1& \mbox{$f_s$ is the second highest score;}\\
p_s^{(2)}: c_3>c_2>c_4>c_5>c_1& \mbox{$f_s$ is the third highest score;}\\
p_s^{(3)}: c_4>c_2>c_3>c_5>c_1& \mbox{$f_s$ is the fourth highest score;}\\
p_s^{(4)}: c_5>c_2>c_3>c_4>c_1& \mbox{else.}
\end{cases}
\end{equation*}

For student with the highest score, since he/she can be assured to get into $c_1$, there is no need to strategize. For other students, since $c_1$ is already occupied, they can shield their first choice and list $c_1$ as their last choice. Thus the student would manipulate
his/her first and second choice as $c_4$ and $c_5$ in order to obtain better opportunity.

\begin{figure}[htb]
  \centering
    \includegraphics[width=0.6\textwidth]{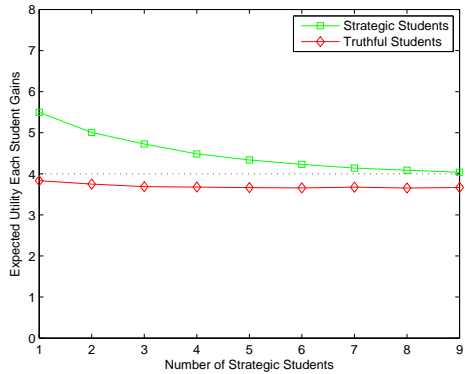}
    \caption{Expected Utility Each Strategic (Truthful) Student Gains over Difference Number of Strategic Students}
    \label{fig:stu_gain_compare_random_alpha}
\end{figure}

Fig. \ref{fig:stu_gain_compare_random_alpha} shows the expected utility of students of either type.
Since the aggregate utility of all ten students are always $40$ under $\beta=1$, the average utility
per student would be $4$ when there are no strategic students, which serves as the benchmark for
evaluating the effectiveness of possible strategies. As we see, when there are only one strategic
student, strategy $\mathbb{S}$ would bring expected utility of $5.5$, much more than the benchmark
utility of $4$ when all students act truthfully. As the number of strategic students increases, the
expected utility per strategic student would decrease gradually, but still remain larger than $4$.
This result validates the effectiveness of the proposed strategy $\mathbb{S}$ under our simulation
setting. The existence of strategic students would make the students who act truthfully worse-off,
however, from the figure we see this side-effect is bounded as the number of strategic students
increase. Even when all other students act strategically, the remaining truthful student could still
achieve about $92\%$ of the benchmark utility.

\begin{figure}[htb]
  \centering
    \includegraphics[width=0.6\textwidth]{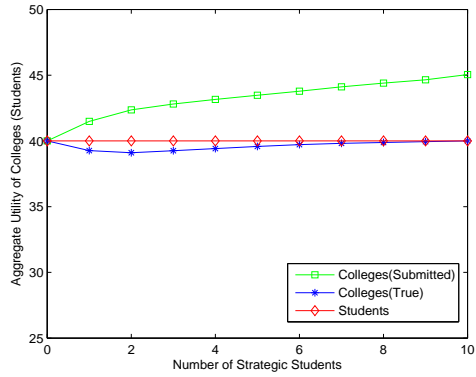}
    \caption{Expected Aggregate Utility of Colleges (Students) over Difference Number of Strategic Students}
    \label{fig:uti_if_strategy_random_alpha}
\end{figure}

Fig. \ref{fig:uti_if_strategy_random_alpha} presents the aggregate utility of both sides when
strategic students exist. The aggregate utility of students remains $40$ since all students share the
same preference list as $R_s$. The aggregate utility of colleges (the square-marked line) is
calculated based on the \emph{submitted} preference lists of students, since in practice the true
preferences of students are unrevealed to the public. When there are no strategic students, all
students submit the same preference list. Thus all colleges would share exactly the same
reciprocating preference over students, ordered solely by the scores of students. This explains
$\pi_C=10+9+8+7+6=40$ when the number of strategic students is zero since student with the highest
score always brings utility of $10$ to colleges, student with the second highest score brings $9$ and
so on. As there are more strategic students, colleges other than $c_1$ would probably enroll students with higher merit scores since strategic students list colleges as more favored choices. Thus the aggregate utility of colleges based on the submitted preferences would increase.

For comparison, we also show the aggregate utility of colleges based on the \emph{true}
preferences of students in the asterisk-marked line. In our setting of $\beta=1$, the true
reciprocating preferences of colleges would be all the same, ranking the students according to their
exam scores. When all students tell truth or all apply strategy $\mathbb{S}$, $c_1$ would enroll the
student with the highest score, $c_2$ gets the student with the second highest score and so on. Thus
$\pi_C^{{t}}=10+9+8+7+6=40$. In other cases, some strategic students with lower scores may receive better
offer than they could achieve when telling truth. For instance, $c_5$ may not receive the student with the
fifth-highest score as anticipated. Therefore the aggregate utility of colleges
would decrease from the initial value of $40$. However, this adverse effect is bounded as seen from
Fig. \ref{fig:uti_if_strategy_random_alpha}, because when there are too many strategic students in
the system, it would be difficult for strategic students with very low scores to achieve better
matching outcome.

Fig. \ref{whole_fig:diff_beta}\subref{fig:stu_uti_stra_vs_true_beta_0} - \subref{fig:stu_uti_stra_vs_true_beta_1} show the
expected utility each particular student gains under \emph{different degrees of preference
correlation}. The horizontal axis denotes the score ranking of student, and each point drawn in the
figure is the average value for one thousand iterations. The possible strategy students could make is strategy
$\mathbb{S}$. We assume there is only one student behaving strategically while the rest of students
would reveal their true preferences. From the figures we see that for the top one student, he/she
would always achieve utility of ten since according to strategy $\mathbb{S}$ he/she would just act
truthfully and can always get the first choice. Fig. \ref{whole_fig:diff_beta}\subref{fig:stu_uti_stra_vs_true_beta_0} verifies
that truth-telling is an equilibrium strategy when students preference are totally uncorrelated:
the students would not be better off when deviating unilaterally from the equilibrium, whereas under
$\beta=1$, truth-telling would be dominated by the strategy $\mathbb{S}$, as shown in Fig.
\ref{whole_fig:diff_beta}\subref{fig:stu_uti_stra_vs_true_beta_1}. For general value of $\beta\in (0,1)$, the students with
higher ranking would tend to act truthfully while the students with lower ranking tend to behave
strategically by choosing the less famous colleges. As the preferences of students becomes more
correlated, more students would have incentive to make certain strategy rather than acting
truthfully.

\begin{figure}[htb]
  \centering
  \subfloat[$\beta=0$]{
  \label{fig:stu_uti_stra_vs_true_beta_0}
  \begin{minipage}{0.47\textwidth}
    \centering
    \includegraphics[width=1.1\textwidth]{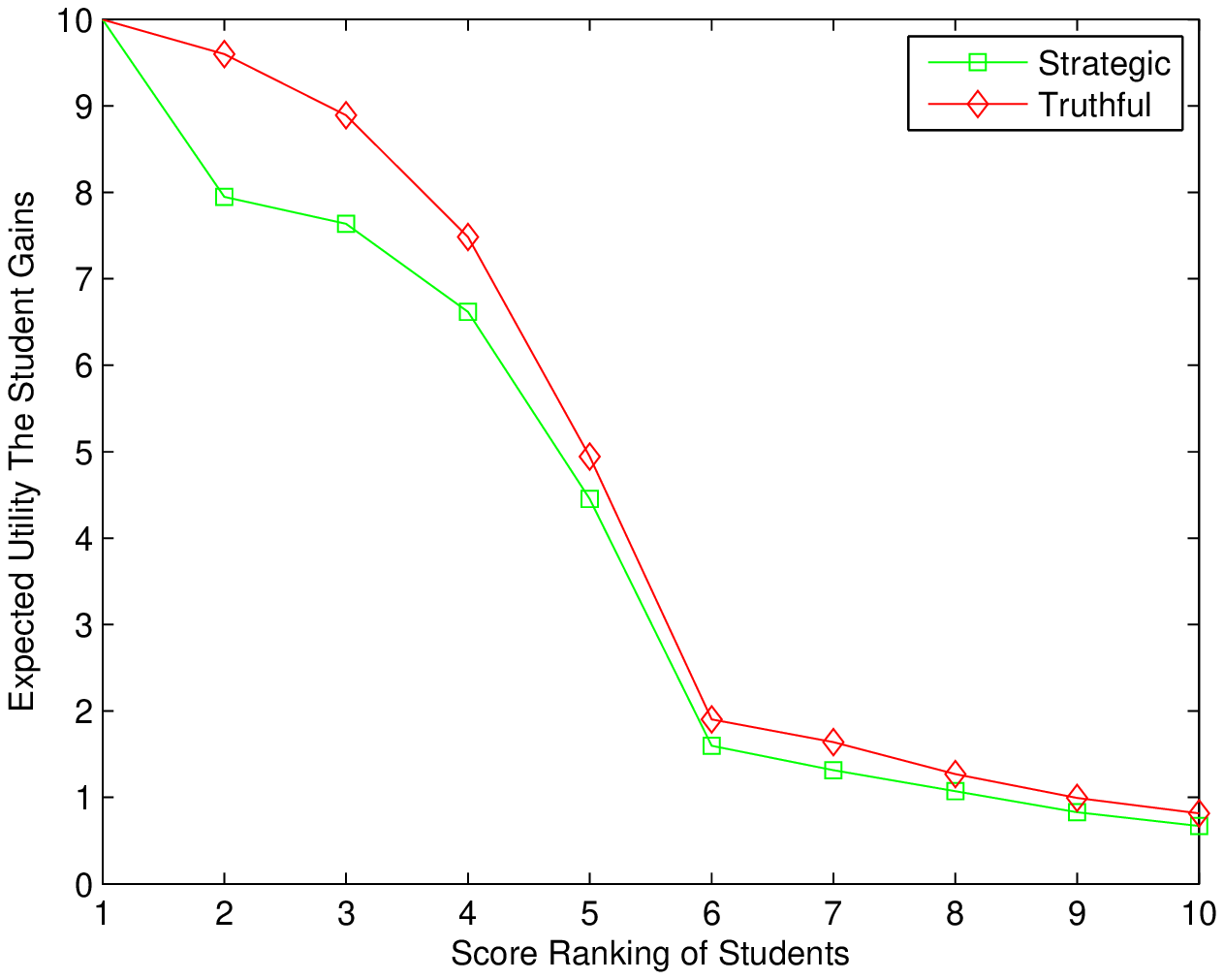}
  \end{minipage}
  }
  \subfloat[$\beta=0.2$]{
  \label{fig:stu_uti_stra_vs_true_beta_0.2}
  \begin{minipage}{0.47\textwidth}
    \centering
    \includegraphics[width=1.1\textwidth]{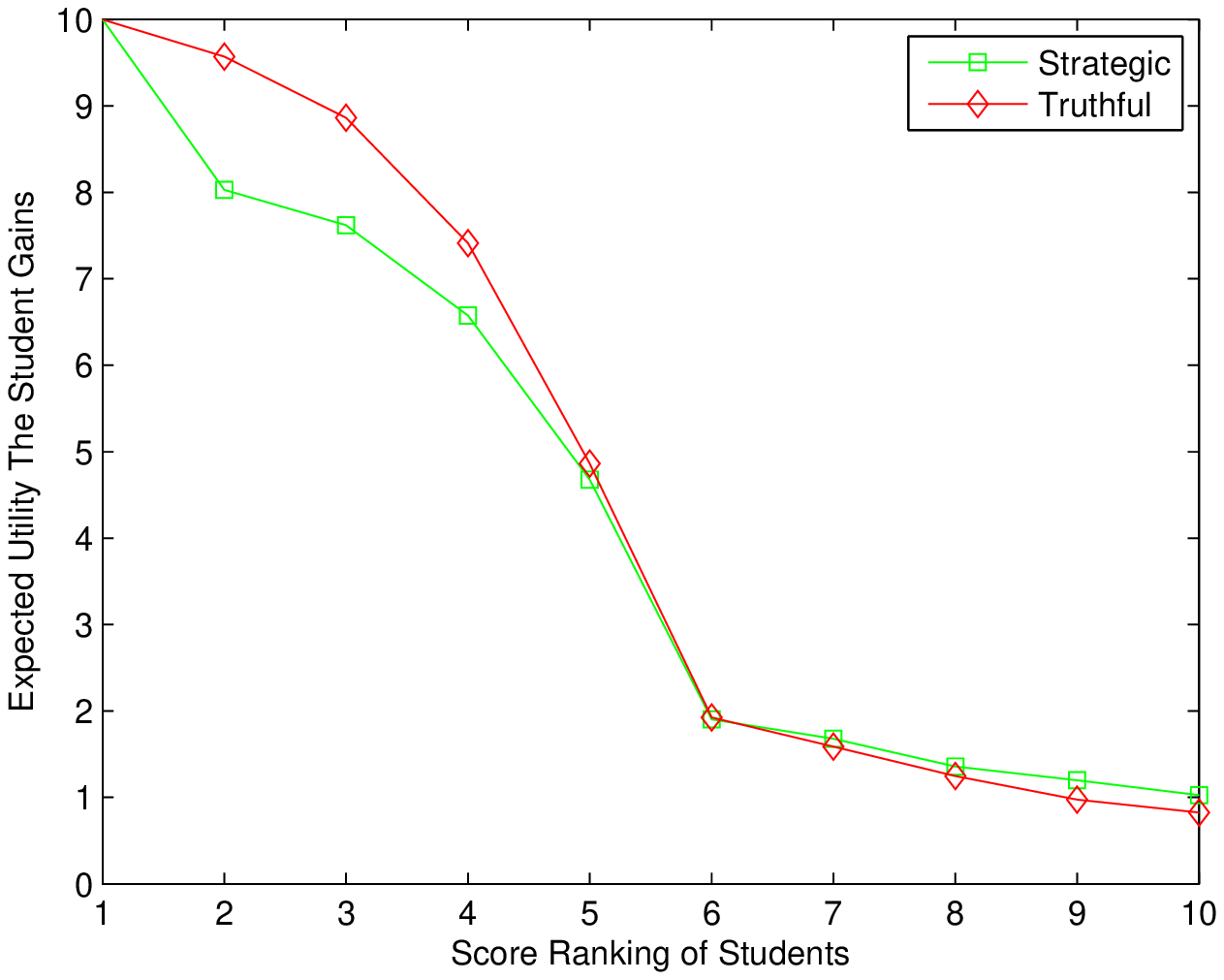}
  \end{minipage}
  }
  \\
  \centering
  \subfloat[$\beta=0.4$]{
  \label{fig:stu_uti_stra_vs_true_beta_0.4}
  \begin{minipage}{0.47\textwidth}
    \centering
    \includegraphics[width=1.1\textwidth]{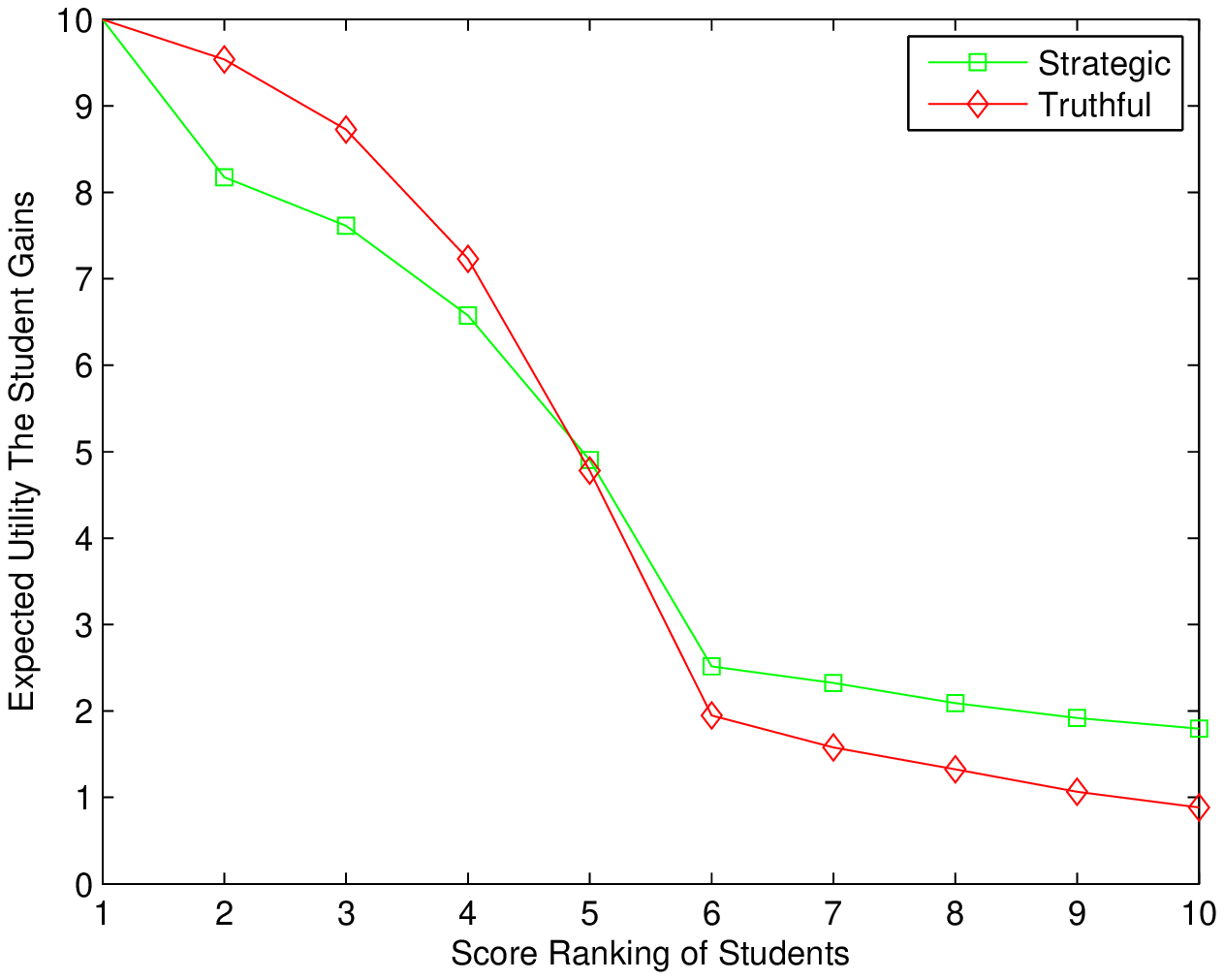}
  \end{minipage}
  }
  \subfloat[$\beta=0.6$]{
  \label{fig:stu_uti_stra_vs_true_beta_0.6}
  \begin{minipage}{0.47\textwidth}
    \centering
    \includegraphics[width=1.1\textwidth]{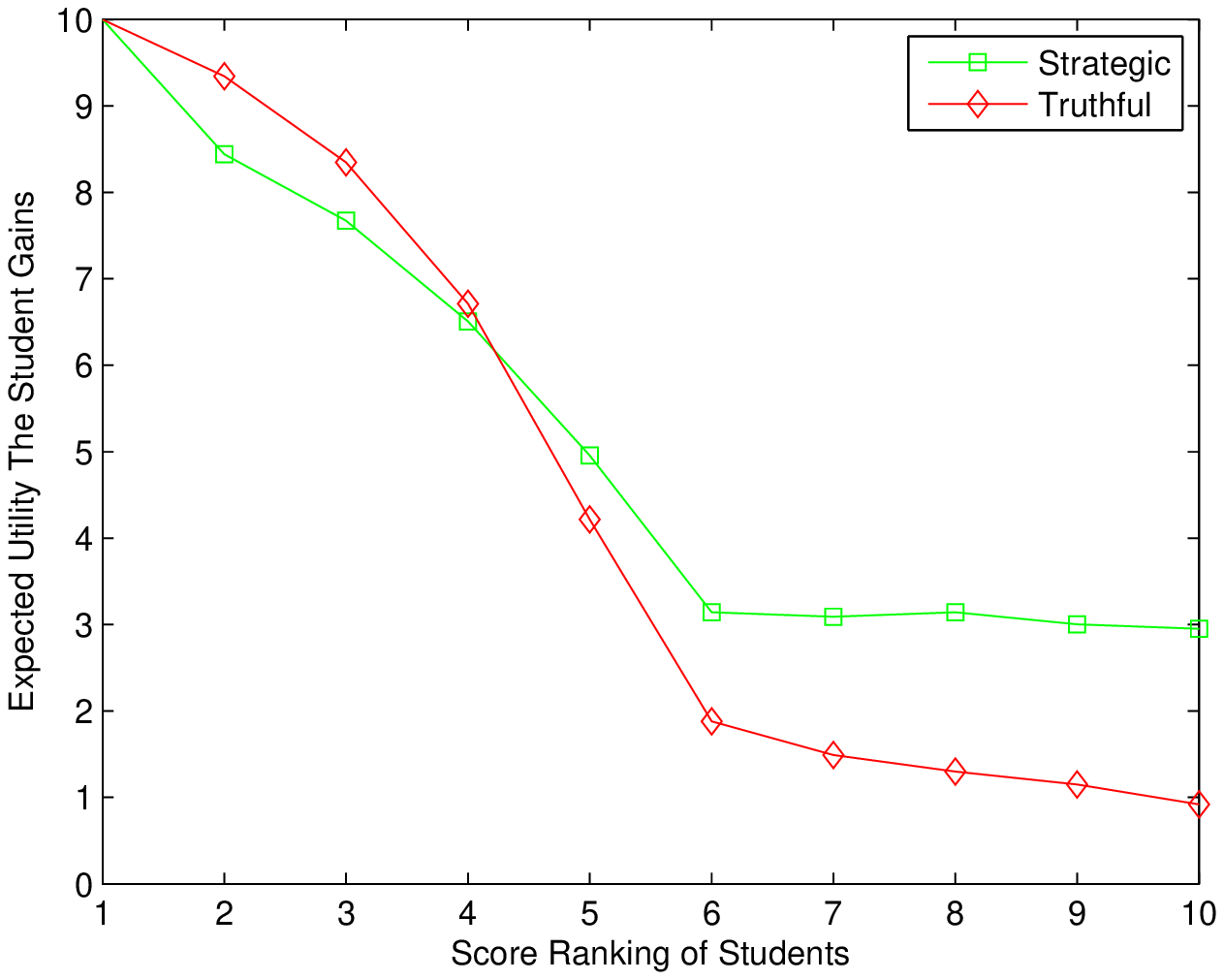}
  \end{minipage}
  }
  \\
  \centering
  \subfloat[$\beta=0.8$]{
  \label{fig:stu_uti_stra_vs_true_beta_0.8}
  \begin{minipage}{0.47\textwidth}
    \centering
    \includegraphics[width=1.1\textwidth]{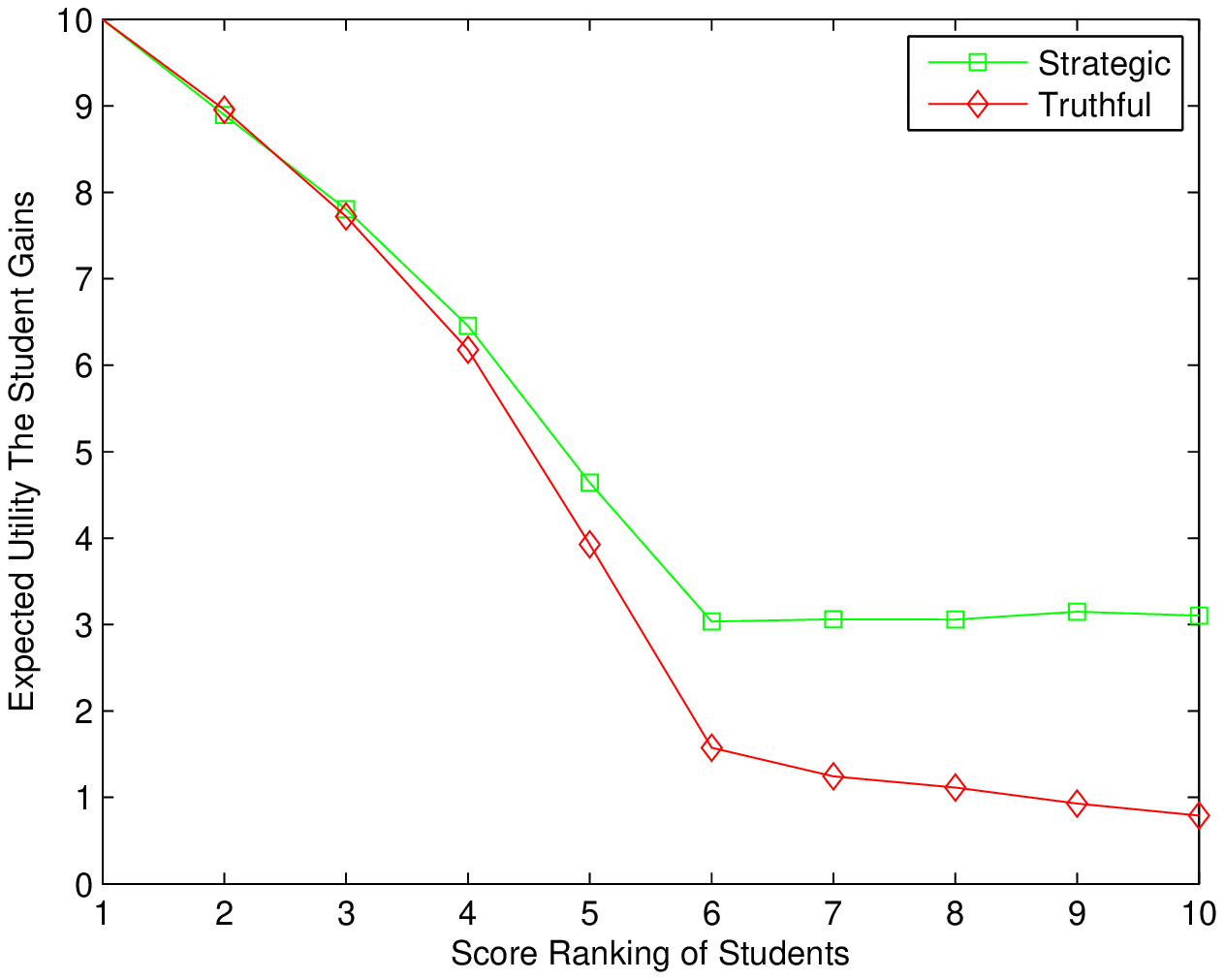}
  \end{minipage}
  }
  \subfloat[$\beta=1$]{
  \label{fig:stu_uti_stra_vs_true_beta_1}
  \begin{minipage}{0.47\textwidth}
    \centering
    \includegraphics[width=1.1\textwidth]{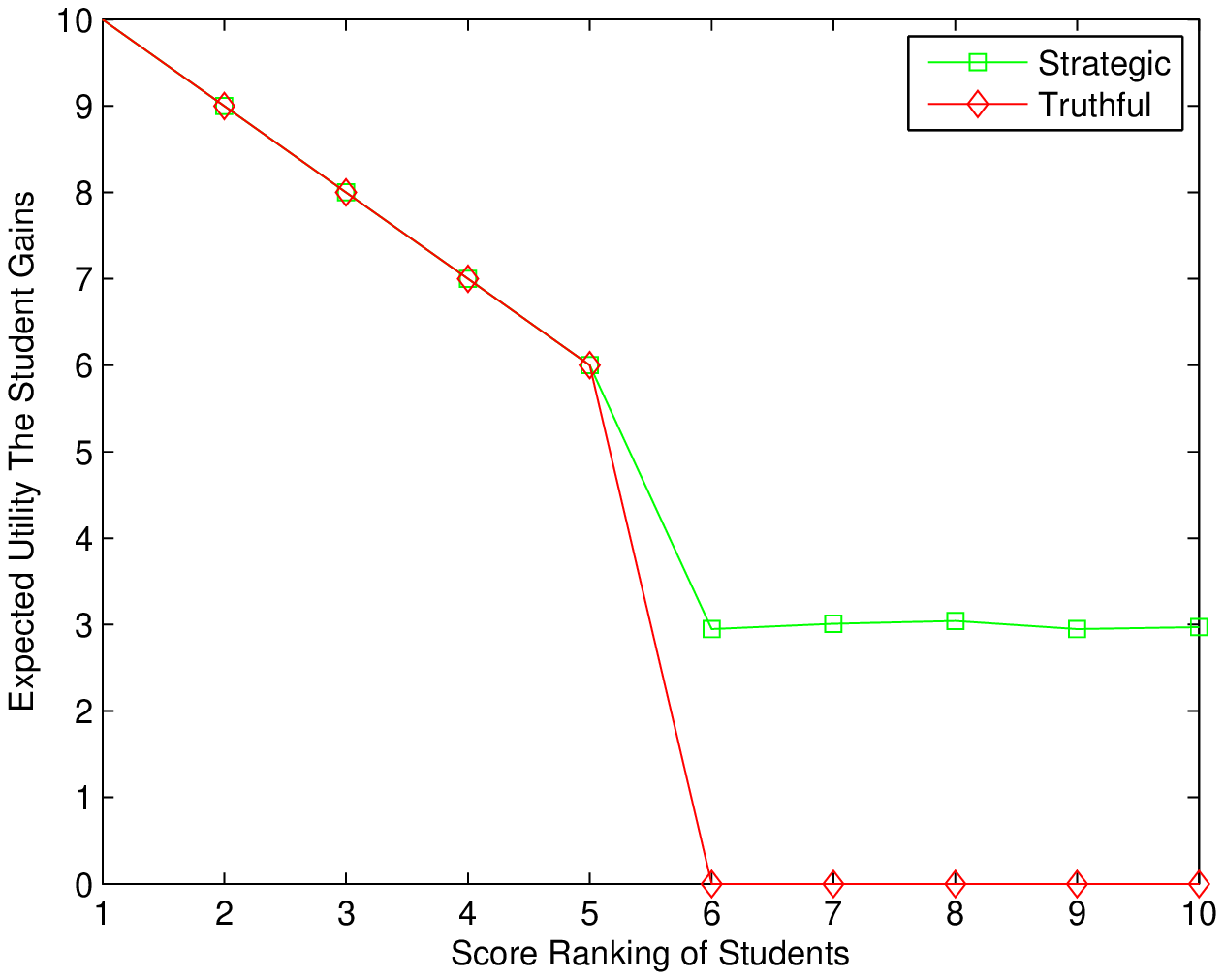}
  \end{minipage}
  }
  \caption{The Expected Utility Each Student Gains under Different Values of $\beta$, Sorted by Score Ranking of Students}
  \label{whole_fig:diff_beta}
\end{figure}

\section{Further Discussions}\label{sec:discussion}
\textbf{Pros and Cons Regarding DA and the Generalized Mechanism}
We have shown that under complete information, students do have strong incentive to lie; on the contrary, in the absence of any posterior knowledge, truth-telling is still the NE of the generalized mechanism. In the simulation, $\beta$ serves as an indicator of how much information each student may know about others' preferences. We would like to highlight through this paper the debate whether college side would all agree on applying DA so as to achieve strategy-proofness property, or some colleges would be willing to enroll students with higher interests in risk of potential manipulation. Our generalized model just establishes a framework to admit such tradeoff between gain in efficiency and risk in truthfulness, and gives each college the freedom to adjust their policy through RF. It is worth noting that the proposed model does not intend to repel the current trend of practice of DA. DA is essentially encompassed in our framework and we just give it a second thought from the perspective of individual college.

To further understand the potential downside of DA, let us consider the following scenario: if DA is enforced in all participating programs of JUPAS, for those less popular programmes, lots of students with lower scores but more interest would be wiped out, and instead most quota of these programmes would be occupied by students with higher scores but less interest (their scores are not high enough to get into popular programmes which they listed as top choices, so finally end up with an offer from bottom choices).\footnote{We notice a recent publication of \cite{ChiuWeng} which shares the same concerns as ours. Instead of considering the hybrid mechanism, \cite{ChiuWeng} rationalizes the so-called ``pre-commit'' strategy of colleges to admit applicants who rank them as their top choices.} One conventional policy in Hong Kong's universities is that students are allowed to switch programme inside the university in the end of the first year in case students feel the current programme is not suitable for themselves after one year's study. In this case most unsatisfied students enrolled by these unpopular programmes would apply to change. If approving most of these applications, those programmes may suffer from high vacancy rate. However, if rejecting most of them, the majority of students in the programme would feel unsatisfactory. This becomes a dilemma for the unpopular programmes. The main objective of DA is to achieve stability; however, it may end up with ``unstable'' outcomes in the long run. This explanation helps rationalize the current selection of hybrid manner in JUPAS rather than hasty replacement with DA. Unfortunately, these concerns and debates from the college side are largely ignored in existing literature.

\textbf{Extension to Marriage Problem}
The generalization of college admissions mechanism also applies to the classic stable marriage
problem discussed in the original paper of \cite{GS}. Consider the situation when a woman
faces two men' proposals and has no clear idea which one she strictly prefers. Technically, we call there exists a \emph{tie} in the woman's preference list.\footnote{Readers could refer to this comprehensive survey for recent
development on the marriage problem in \cite{Survey_Marry_prob}, especially section on ``incomplete
preference lists with ties''.} Roughly speaking, the existing literature mainly provides two
solutions to deal with marriage problem with tie. A quick solution is just requiring the woman to flip
a coin to produce a strict preference list so that the previous mechanism could be applied
immediately. The other solution concerns how to find the optimal matching outcome among all these
artificial tie-breaking possibilities, for instance, the polynomial-time stable improvement cycles
algorithm raised in \cite{Tie-breaking}.

However, suppose the woman has the wish: ``I'll choose the man who loves me most!'' Yet
the fact is that the first man has listed her as the first choice while the second man listed this
woman as his last choice and was rejected by every other women in the previous rounds, assuming
men-proposing DA algorithm \citep{GS} applied here. Obviously a ``reasonable'' matching mechanism
should respect each participating agent' wish and therefore always match the first man with the
woman. However, existing matching mechanisms provide no channel for agents to express such kind of
reciprocating preferences, although mutual appreciation is a very natural and common factor in
determining marriage mates.

Our generalized model can be easily carried over to the above scenario. By analogy, we can define the
\emph{merit score} of man $m$ in woman $w$ as follows:\footnote{Similarly, by exchanging the notation
of $m$ and $w$ in the equation, we may generate the merit score of woman $w$ in man $m$.}
\begin{equation*}
    {mrt}_w(m)=(1-\alpha_w)\cdot f_w(m)+\alpha_w\cdot h_w(r),\quad \alpha_w\in [0,1]
\end{equation*}
where $f_w(m)$ is the initial rating score of $m$ in $w$; $r$ is an integer denoting $w$'s position
in $m$'s initial preference list and $h_w(r)$ is the bonus score which is decreasing over $r$, i.e.,
the smaller $r$ is, the more bonus score the man can achieve; $\alpha_w$ is the reciprocating factor
of the woman, denoting $w$'s sensitivity to other men's evaluation to herself. If $\alpha_w$ is set
to zero, which is equivalent to the classic model where a woman only believes her own feeling and
judgement. Conversely, if $\alpha_w$ equals to one, woman $w$ is extremely sensitive to men' opinions
on her and hopes to match with the one who loves her most. In general a woman may set $\alpha_w$
between zero and one to strike a balance between her initial feeling and men's appraisal to her.
Choosing a mate with mutual appreciation seems more ``reasonable'' and natural in practical marriage.
The parameter of \emph{reciprocating factor} provides an opportunity for agents to more fully express
their wishes or perceived payoff than in the classic model. Finally by comparing the \emph{merit scores} for different
men, we can reproduce the preference list of $w$. In case of tie in merit scores, we resort to
original score for tie breaking.

\textbf{Related Work}
Some recent developments in matching theory share our concerns that the transition from the Boston
mechanism to the GS mechanism is not problem-free and object to the hasty rejection of the Boston
mechanism. One main research direction is the analysis of efficiency in school choice (SC) setting
where schools do not have \emph{strict} preferences over students and have to largely rely on
\emph{random lotteries} to determine their preferences.

Abdulkadiro\u{g}lu, Che and Yasuda first brought up the uncertainty factor of lotteries into
efficiency consideration \citep[see][]{Expanding}. They showed an elegant example when students share
identical ordinal preference but differ in preference intensities, the Boston mechanism can dominate
the GS mechanism in terms of \emph{expected} cardinal efficiency. A new Choice-Augmented Deferred
Acceptance (CADA) mechanism was proposed accordingly which supports a greater scope of efficiency
than the pure GS mechanism \citep[see][]{Expanding}. The same authors further generalized the single example
into a ``baseline model'' where students have common ordinal preferences and schools have no
priorities in \cite{Boston_reconsi}. Besides, \cite{SCcase} showed that the above
analytical results could extend to more realistic cases such as weak priorities by simulation.

Featherstone and Niederle then classified the efficiency issue in SC into three categories
\citep[see][]{Ex_ante}:

\begin{description}
  \item[Ex post:] Each student knows preferences of other students and lottery results in each school.
  The matching outcome as well as the efficiency are both deterministic.
  \item[Interim:] Students know preferences of other students but remain unknown to the lottery
  results, i.e., we would investigate the efficiency before the lotteries are drawn. The distribution
  of lottery results would induce an expected, other than deterministic, value of efficiency.
  \item[Ex ante:] Students only know the distribution of other students' preferences and still remain ignorant
  of lottery results.
\end{description}

The authors concluded in the same paper that, when student preferences are uniformly distributed and
schools are completely symmetric, the Boston mechanism can first-order stochastically dominate the GS
mechanism in terms of \emph{ex ante} efficiency, both in theory and in the laboratory.

Following the efficiency classification above, results in \citep{Expanding,Boston_reconsi,SCcase}
would all fall into the \emph{interim} viewpoint with highly \emph{correlated} student preferences,
which complements the conclusion in \cite{Ex_ante} under \emph{independent} student preferences.

Although sharing the same caution against a hasty replacement of the Boston mechanism, our paper
stands distinct from these above works in several aspects:
\begin{itemize}
  \item One \emph{key} assumption for the above works is the \emph{weak} or even no priorities in schools such that
  lotteries are largely relied upon in schools in order to break the tie. It is this ``randomness'' that causes the potential
  \emph{ex ante} efficiency loss of the GS mechanism. However, in practical CA context where students'
  scores rather than the random lotteries play the decisive role in admissions, the above assumption
  would no longer hold, so would the corresponding conclusions.
  \item Our paper follows a distinctive and unique research direction and shows that even when the priorities in
  schools are \emph{strict}, Boston still exhibits some prominent properties such as respecting the interests
  of applicants. The sociological consideration of agents' preferences has been largely ignored in previous research of college admissions system.
\end{itemize}

In terms of interdependent preferences, we notice a recent work of \cite{Interdependent} proposing
``interdependent values'' in two-sided matching which can be regarded as a complementary notion to
our reciprocating preferences. In \cite{Interdependent} the authors argue that a college $c$'s
evaluation of a student $s$ could be affected by (or depend on) other colleges' evaluation to this
student $s$, while we consider the scenario where $s$'s value to $c$ is dependent on $c$'s value to
$s$.

\section{Concluding Remarks}\label{sec:conclusion}
In this paper we propose a generalized matching mechanism which can incorporate both BM and DA.
Inspired by a practical college admissions system, i.e., JUPAS in Hong Kong, we propose a
common parameter, namely \emph{reciprocating factor} (or $\alpha$), for the generalized matching model.
This parameter serves as a bridge between BM and DA: when all $\alpha$ are set to zero, the matching
mechanism would be equivalent to pure BM; when all $\alpha$ equal to one, the matching mechanism
reduces to pure DA. Practical systems like JUPAS can be regarded as a hybrid of BM and DA with
reciprocating factor between zero and one. In the context of college admissions, reciprocating factor
is of practical significance for programmes to achieve the tradeoff between eligibility and real interest of enrolled students. We have discussed the advantage and disadvantage of DA and the generalized mechanism extensively and highlighted the debate from the perspective of colleges side. These potential concerns and doubts from colleges would help justify the current selection of hybrid system in Hong Kong.

With regard to future works, our paper could be further improved and extended in several directions:
\begin{enumerate}
  \item One major open question regarding the generalized mechanism is how history data of admission may affect the strategies and choices of current applicants. Are there any equilibria for this extensive-form game that all participants would conform to? How would their behaviors in the system evolve in the long run? We hope that these strategic issues can be tackled in the future.
  \item How the reciprocating factors of colleges are distributed is another interesting direction which we have not yet investigated in details. \cite{ChiuWeng} showed that both popular and unpopular colleges have motives of giving preferential treatment to applicants who rank them as top choices. Does this imply that different colleges tend to have similar reciprocating factors in practice? This question may not be addressed immediately since these ``inside'' information of colleges is typically inaccessible to the public.
  \item We have showed some positive results of the hybrid system through simulation, however, simulation alone is insufficient to cover all complicated strategies of students. As the future work, we plan to design some lab experiments which would involve real participants to play and learn during the game. We believe that these empirical results would provide more evidence and insight for supporting and spreading the adoption of hybrid system like JUPAS in Hong Kong.
\end{enumerate}

\appendix
\section*{Appendices}
\section{Proof of Property \ref{theo:no_truthful_mech}}\label{proof:no_truthful_mech}
It is easy to see that when all RF are zero, the generalized mechanism reduces to DA which is strategy-proof. To prove a mechanism is non-truthful in general, we just need to find one counter-example for any given RF which are not all zero.

Suppose all colleges have zero RF except college $c$ with positive $\alpha_c$. Besides, $c$ has quota for only one student. Let $\alpha_c$ equal to any small $\epsilon>0$. Suppose that there are two students $s_1$ and $s_2$. $s_1$ achieves a higher score $f_{s_1}$ and $c$ is his/her second choice school; $s_2$ gets a lower score $f_{s_2}$ but lists $c$ as his/her first choice. The merit score of $s_2$ in $c$ is then:
\begin{equation*}
    mrt_c(s_2)=(1-\epsilon)\cdot f_{s_2}+\epsilon \cdot h_c(1)
\end{equation*}
If telling truth, $s_1$'s merit score in $c$ would be:
\begin{equation*}
    mrt_c(s_1)=(1-\epsilon)\cdot f_{s_1}+\epsilon \cdot h_c(2)
\end{equation*}
By letting $mrt_c(s_1)<mrt_c(s_2)$, we have the condition when $c$ prefers $s_2$ to $s_1$:
\begin{equation*}
  f_{s_1}-f_{s_2}<\frac{\epsilon}{1-\epsilon}(h_c(1)-h_c(2))\triangleq\delta
\end{equation*}
Since $0<\epsilon<1$ and $h_c$ is a (strictly) decreasing function of choice order, we have $\delta>0$.

We further assume that $s_1$ knows he/she would not succeed in his/her first choice college, thus the best he/she can achieve is getting into college $c$. By listing $c$ as the first choice, $s_1$ would obtain a higher merit score in $c$:
\begin{equation*}
    mrt_c(s_1)^S=(1-\epsilon)\cdot f_{s_1}+\epsilon \cdot h_c(1)
\end{equation*}
which is always larger than $mrt_c(s_2)$.

Therefore, as long as $f_{s_1}-f_{s_2}<\delta$, $s_1$ would have incentive to deviate from truth-telling for better matching outcome, which proves that the hybrid mechanism is not strategy-proof in general.

\section{Proof of Property \ref{theo:Bayes_NE}}\label{proof:Bayes_NE}
Here we follow the terms in \cite{Ex_ante} and adapt their proof for BM to our context of generalized mechanism. We first give some basic definitions and lemmas.

\begin{definition}
For student $s$, let $p$ be its preference list where college $c$ is ranked $r$-th. Consider a new preference list $p'$ of $s$ which exchanges the $r$-th choice with the $j$-th ($j<r$) choice college. A mechanism is called \emph{rank monotonic} if the probability of $s$ being matched to $c$ is weakly higher under $p'$ than under $p$ no matter how preferences and reciprocating factors are distributed.
\end{definition}

\begin{lemma}\label{lemma:rank_mono}
The generalized mechanism is \emph{rank monotonic}.
\end{lemma}
\begin{proof}
Consider the notation in the previous definition. Also consider any state of the matching system (i.e., any possible scores, submitted preferences of students and RF of colleges) where student $s$ submits $p$ and is matched to college $c$. In this state, student $s$ is rejected by every college prior to its $r$-th choice in $p$, and after $c$ enrolls $s$, no other students can drive $s$ out of $c$'s quota in later rounds. Mathematically, suppose $c$ has quota $q$ and $s$ is in the $i$-th position in $c$'s merit order list, then there will be at most $q-1$ students with higher position than $s$ in $c$'s list.

Hence, if student $s$ had instead submitted $p'$, in the same state of the system, he also would have been rejected by every college prior to its $j$-th choice in $p'$. Considering $c$'s merit order list, student $s$ can only get promoted since it lists $c$ as higher choice in $p'$. After getting into $c$ in the $j$-th round, since preferences of all other students and colleges remains the same as under $p$, there will also be at most $q-1$ students with higher position than $s$ in $c$'s new list. So $s$ can guarantee its admission to $c$. This implies that the probability of student $s$ being matched to college $c$ is weakly larger when it submits $p'$ instead of $p$.  \qed
\end{proof}

\begin{definition}
If the probability of student $s$ being matched to its $i$-th choice college is independent of its submitted preference list, we say that its preference revelation problem exhibits \emph{college anonymity}.
\end{definition}

\begin{lemma}\label{lemma:college_anonym}
Suppose there are $m$ colleges with quota $q$ and $n$ students. If the submitted preferences of all students other than $s$, as well as the reciprocating factors of colleges are uniformly distributed, then the preference revelation problem of student $s$ exhibits \emph{college anonymity}.
\end{lemma}

\begin{proof}
Let student $s$ submit preference list $p_s$ where $p_s(i)$ is its $i$-th choice college in list $p_s$. Denoted by $p_{-s}$ other student's preferences and $\vec{\alpha}=(\alpha_1,\alpha_2,\ldots,\alpha_m)$ the vector of each college's reciprocating factor. Let $\mathcal{P}$ be the set of all possible $p_{-s}$. And $\mathcal{R}^m=[a,b]^m$ is the $m$-dimensional space containing all possible values of $\vec{\alpha}$.

For each particular $p_{-s}\in\mathcal{P}$, denoted by $\mathcal{R}_{p_s,p_{-s},i}^m(\alpha_1,\alpha_2,\ldots,\alpha_m) \subseteq \mathcal{R}^m$ the region of $\vec{\alpha}$ in which $s$ is matched to its $i$-th choice under $(p_s,p_{-s})$. Then the probability of $s$ being matched to its $i$-th choice can be denoted as:
\begin{equation}\label{eqn:prob_i}
    prob_i=\sum_{p_{-s}\in\mathcal{P}} Pr(p_{-s})\int_{\mathcal{R}_{p_s,p_{-s},i}^m}\frac{1}{(b-a)^m}d\vec{\alpha}
\end{equation}

Now let student $s$ submit a different preference list $p_s'$. This induces a permutation mapping $f$ from the set of colleges' index $\{1,2,\ldots,m\}$ to itself which is defined as follows: for any college $c_j$ $(j\in\{1,2,\ldots,m\}$, let $r_{c_j,p_s}$ be its ranking in list $p_s$, then the image value $f(j)$ is the index of college which is ranked $r_{c_j,p_s}$ in list $p_s'$. By a slight abuse of notation, we express this permutation as $f(p_s)=p_s'$ for conciseness.

For each $(p_{-s},\mathcal{R}_{p_s,p_{-s},i}^m(\alpha_1,\alpha_2,\ldots,\alpha_m))$, by symmetry, we know that under \\ $(f(p_{-s}),\mathcal{R}_{p_s,p_{-s},i}^m(\alpha_{f(1)},\alpha_{f(2)},\ldots,\alpha_{f(m)}))$, student $s$ would still be matched to its $i$-th choice college in $p_s'$. Since $f$ is a one-to-one mapping, by this way of transformation we would have listed all the possibility for the matching. Thus the probability of $s$ being matched to its $i$-th choice in $p_s'$ can be denoted as:
\begin{equation}\label{eqn:prob_i_prime}
  prob_i' = \sum_{p_{-s}\in\mathcal{P}} Pr(f(p_{-s}))\int_{\mathcal{R}_{p_s,p_{-s},i}^m(f(\vec{\alpha}))}\frac{d(f(\vec{\alpha}))}{(b-a)^m}
\end{equation}
Since preference list is uniformly drawn, we have $Pr(p_{-s})=Pr(f(p_{-s}))$. The integral parts in equation (\ref{eqn:prob_i}) and (\ref{eqn:prob_i_prime}) are equivalent since it is just a simple substitution of the variables, i.e., they both equal to the following expression:
$$\int_{\mathcal{R}_{p_s,p_{-s},i}^m(x_1,x_2,\ldots,x_m)}\frac{1}{(b-a)^m}dx_1dx_2\ldots dx_m$$

Thus, we have $prob_i=prob_i'$ and the lemma is proven.  \qed
\end{proof}

We still use $prob_i$ to denote the probability of $s$ being matched to its $i$-th choice college. Combining the conclusion in Lemma \ref{lemma:rank_mono} and Lemma \ref{lemma:college_anonym}, we have that $prob_i\geq prob_j$ for $i<j$. The best response in this case is clearly to put the favorite school in the first place, the second most favorite one in the second place and so on. Thus truth-telling would be the best response for students.

\section{Proof of Property \ref{theo:dominated}}\label{proof:dominated}
Suppose there are $m$ colleges and $n$ students. Without loss of generality, we assume that all students share the same preference list: $c_1>c_2>\cdots>c_m$. Let the quota of their favorite college $c_1$ be $q$ and reciprocating factor of $c_1$ is $\alpha_{c_1}$.

By listing $c_1$ as their first choice, the top $q$ students judged by score (denoted by $s_1,s_2,\ldots,s_q$ with $f_{s_1}> f_{s_2}> \ldots> f_{s_q}$) would still occupy the top $q$ vacancies in $c_1$'s merit score list. This can be easily seen by the composition of merit score:
\begin{equation*}
    mrt_{c_1}(s_q)=(1-\alpha_{c_1})\cdot f_{s_q}+\alpha_{c_1} \cdot h_{c_1}(1)
\end{equation*}
For any other student $s_j$ ($q<j\leq n$) listing $c_1$ as $i$-th choice, its merit score would be smaller than $mrt_{c_1}(s_q)$ since $f_{s_j}<f_{s_q}$ and $h_{c_1}(i)\leq h_{c_1}(1)$.

Thus, the top $q$ would reveal their true preferences and get into their first choice college.

For the rest students, since all the quota of $c_1$ is filled up, they would have no incentive to still list $c_1$ as their first choice. To prove that truth-telling (denoted by $\mathbb{T}$) is a dominated strategy, we only need to find one particular strategy which dominates $T$. We now focus on one of such strategies which requires students to always submit $c_2>c_3>\cdots>c_m>c_1$ (denoted by $\mathbb{S}$). Let $s$ be any of these remaining students and the utility of $s$ be $u_i$ when it is matched with $c_i$ ($i\in\{1,2,\ldots,m\}$). Denoted by $prob_i$ the probability of $s$ being matched with $c_i$ under $\mathbb{T}$ and $prob_i'$ the same probability under $\mathbb{S}$ (notice that $prob_1$ and $prob_1'$ would be zero since $c_1$'s quota has already been occupied). Then the expected utility under $\mathbb{T}$ would be no more than that under $\mathbb{S}$ since:
\begin{eqnarray*}
    E(U)^T&=&\sum_{i=1}^m prob_i \cdot u_i \\
    &=& \sum_{i=2}^m prob_i \cdot u_i \\
    &\leq& \sum_{i=2}^m prob_i' \cdot u_i \\
    &=& E(U)^S
\end{eqnarray*}
According to Lemma \ref{lemma:rank_mono}, the probability of $s$ being matched with $c_i$ ($i\in\{2,3,\ldots,m\}$) would be higher under $\mathbb{S}$ since $c_i$'s position gets promoted compared with $\mathbb{T}$. Thus $prob_i \leq prob_i'$ and the inequality above holds. This gives the proof that truth-telling is a dominated strategy.

\section{Proof of Property \ref{theo:college_truth}}\label{proof:college_truth}
For analyzing colleges' strategies, we borrow the idea of \emph{dropping strategies} and \emph{rejection chains algorithm} from \cite{Incentive} and show that it is impossible for colleges to find an effective dropping strategy and manipulate the matching successfully. Thus the best matching colleges may achieve is via revealing their true preferences (i.e., their merit order lists).

To complete our proof, we first briefly restate some basic definitions and lemmas in \cite{Incentive}. Let $p_c$ be the reciprocating preference of college $c$ obtained through its true $\alpha_c$ and $h_c(r)$. $s_1$ and $s_2$ are two students applying to college $c$.

\begin{definition}\label{defn:dropping}
A report $p_c'$ is said to be a \emph{dropping strategy} if (i) $p_c: s_1>s_2$ and $p_c':s_1>\varnothing$ imply $p_c':s_1>s_2$, and (ii) $p_c: \varnothing>s_1$ implies $p_c': \varnothing>s_1$.
\end{definition}
In other words, a dropping strategy of a college is obtained by removing some students from its true lists of acceptable students, which never changes its relative preference of any two students.

\begin{lemma}\label{lemma:dropping}
\emph{(Dropping strategies are exhaustive)}: If under certain strategy $p_c'$ the mechanism produces matching outcome $\mu$, then there must exist a dropping strategy of college $c$ that produces a matching that $c$ weakly prefers to $\mu$ according to its true preference $p_c$.
\end{lemma}
The detailed proof can be found in Appendix B of \cite{Incentive}. Lemma \ref{lemma:dropping} implies that if there exists any successful strategy $p_c'$ for college $c$ (by manipulating its $\alpha_c$ and $h_c(r)$), we can always find a dropping strategy which achieves at least the same improvement. That is to say, if we can prove that for every dropping strategy, we cannot make college $c$ better off, then $c$ would be forced to act truthfully.

For any dropping strategies, the ``dropped'' students in $c$ would re-apply to other colleges in the student-proposing DA algorithm, which would cause a chain of rejection and acceptance in the subsequent stages. This can be analyzed in details through the so-called \emph{rejection chains} algorithm as follows.

\noindent\textbf{Algorithm 1.} \textsc{Rejection Chains}

\noindent Let $\mu$ be the outcome of student-proposing DA algorithm and $B_c^1$ a subset of $\mu(c)$. Let $c$ reject all the applicants in $B_c^1$. Initially, we set $i=0$ and $flag=false$.

\noindent\textbf{BEGIN}

\noindent Let $i:=i+1$;

\begin{enumerate}
  \item If $B_c^i=\varnothing$, \textbf{return};
  \item Otherwise, let $s$ be the least preferred student by $c$ among $B_c^i$, and let $B_c^{i+1}:=B_c^i\backslash s$.
  \item Iterate the following steps.
  \begin{enumerate}
    \item $s$ continues to apply:
    \begin{enumerate}
      \item If $s$ has already applied to every college in $p_s$, \textbf{GO TO} BEGIN;
      \item Otherwise, let $c'$ be the most preferred college of $s$ among those which $s$ has not yet applied. If $c'=c$, set $flag=ture$ then \textbf{return};
    \end{enumerate}
    \item Acceptance and/or rejection:
    \begin{enumerate}
      \item If $c'$ has no vacant position and prefers each of its current mates to $s$, then $c'$ rejects $s$, \textbf{GO TO} Step 3.
      \item Otherwise, $c'$ accepts $s$. If $c'$ has a vacant position, \textbf{GO TO} BEGIN; Otherwise, $c'$ rejects the least preferred student among those who were matched to $c'$. Let this rejected student be $s$ then \textbf{GO TO} Step 3;
    \end{enumerate}
  \end{enumerate}
\end{enumerate}

\noindent\textbf{END}

When Algorithm 1 returns, $flag$ could be either $false$ or $true$. We say that Algorithm 1 \textbf{returns to $c$} if it returns with $flag=true$.

We then re-state Lemma 3 in \cite{Incentive} as follows:
\begin{lemma}\label{lemma:reject_chain}
For any $c\in\mathcal{C}$, if Algorithm 1 does not return to c for any non-empty $B_c^1\subseteq \mu(c)$, then $c$ cannot profitably manipulate by a dropping strategy.
\end{lemma}

Therefore the proof can be now boiled down to answer whether Algorithm 1 can possibly return to $c$ after applying any dropping strategy. This is summarized in the following lemma as the final step of proof.
\begin{lemma}\label{lemma:no_return}
In the generalized mechanism, Algorithm 1 can never return to college $c$ with strictly better $\mu(c)$ if it applied any dropping strategies.
\end{lemma}
\begin{proof}
Assuming that Algorithm 1 returns to $c$ with strictly better $\mu(c)$ after finite steps of iterations. By retrieving the running process of Algorithm 1, we would be able to reconstruct a complete chain\footnote{To be exact, it may be called a rejection ``cycle'' since the starting point of the chain coincides with its ending point.} of rejection and acceptance as follows (Denote $c=c_1=c_n$):
\begin{equation*}
    c_1-s_1-c_2-s_2-\cdots-c_{n-1}-s_{n-1}-c_n
\end{equation*}
The above chain reflects the latest iteration before Algorithm 1 returns with $flag=true$. It starts from $c_1$ (i.e., $c$) rejecting $s_1$, then $s_1$ applied to $c_2$ following its own preference $p_{s_1}$; after comparing $s_1$ and $s_2$, $c_2$ decided to accept $s_1$ and reject $s_2$; then $s_2$ started its new application and so on. Finally, $s_{n-1}$ was rejected by $c_{n-1}$ and applied to $c_n$ (i.e., $c$), which satisfies the stopping condition of Algorithm 1 with $flag=true$.

By inspecting the behavior of $c_2$, we know that $s_2$ applied to it earlier than $s_1$ (which implies $s_2$ obtained higher bonus score than $s_1$ in $c_2$) but is still less favored by $c_2$. The only explanation for this phenomenon is that $s_1$ achieved strictly better exam score than $s_2$, i.e., $f_{s_1}>f_{s_2}$. This rule also applies to other colleges in this chain. In the end, since we suppose $c$ achieves better matching through the dropping strategy, $s_{n-1}$ is regarded better than $s_1$ for $c$, which implies that $f_{s_{n-1}}>f_{s_1}$. In summary, we have:
\begin{eqnarray*}
  f_{s_1} &>& f_{s_2} \\
  f_{s_2} &>& f_{s_3} \\
   &\cdots&  \\
  f_{s_{n-2}} &>& f_{s_{n-1}} \\
  f_{s_{n-1}} &>& f_{s_1}
\end{eqnarray*}
By adding the left side and right side of those inequations, we derive $0>0$, an apparent contradiction!

Thus our initial assumption cannot be true and colleges may not manipulate their preference successfully.  \qed
\end{proof}



\bibliographystyle{spr-chicago}

\end{document}